\documentclass[twocolumn,english]{svjour3}
\usepackage[T1]{fontenc}
\usepackage[latin9]{inputenc}
\usepackage{url}
\usepackage{amsmath}
\usepackage{amssymb}
\usepackage{breqn}

\makeatletter
\newenvironment{svmultproof}{\begin{proof}}{\qed\end{proof}}

\RequirePackage{fix-cm}

\smartqed  

\makeatother

\usepackage{babel}

\begin{document}
\title{From Thermodynamic Sufficiency to Information Causality}
\titlerunning{From Sufficiency to Causality}
\author{Peter Harremo{\" e}s}
\authorrunning{P. Harremo{\" e}s}
\institute{P. Harremo{\" e}s \at Copenhagen Business College\\
N{\o}rre Voldgade 34\\
K{\o}benhavn K\\
Denmark\\
Tel.: +45-39564171\\
\email{harremoes@ieee.org}\\
}
\date{Received: date / Accepted: date}
\maketitle
\begin{abstract}
The principle called information causality has been used to deduce
Tsirelson's bound. In this paper we derive information causality from
monotonicity of divergence and relate it to more basic principles
related to measurements on thermodynamic systems. This principle is
more fundamental in the sense that it can be formulated for both unipartite
systems and multipartite systems while information causality is only
defined for multipartite systems. Thermodynamic sufficiency is a strong
condition that put severe restrictions to shape of the state space
to an extend that we conjecture that under very weak regularity conditions
it can be used to deduce the complex Hilbert space formalism of quantum
theory. Since the notion of sufficiency is relevant for all convex
optimization problems there are many examples where it does not apply.
\keywords{Bregman divergence \and multipartite system \and information causality
\and thermodynamic sufficiency} 
\PACS{03.65.Ud \and 03.67.-a \and 03.65.Aa} 
\subclass{81P16 \and 94A17}
\end{abstract}

\section{Introduction}

Entanglement is a resource that may allow agents to solve certain game
problems in a more efficient way than what is possible without entanglement.
Such tasks could be solved even more efficiently if the agents had
access to a fictive resource called PR-boxes. Such boxes cannot be
used for signaling, but they can create correlations that are stronger
than the correlations that can be created using entanglement. To be
more precise, all quantum mechanical correlations satisfy Tsirelson's
bound while PR-boxes can violate Tsirelson's bound. 

The goal is to explain Tsirelson's bound and other bounds on correlations
from more basic physical principles. One such principle is called
information causality, and it may be formulated as ``one bit of communication
cannot create more than one bit of correlation''. In \cite{Pawlowski2009}
this principle was introduced and it was proved that it can be used
to derive Tsirelson's bound. In \cite{Pawlowski2009} information
causality was formulated and derived from the existence of the function
conditional mutual information that is assumed to satisfy some basic
properties. In \cite{Short2010} two ways of defining entropy were
specified, and they were used to formulate the principle of information
causality. 

In this paper use properties of Bregman divergences rather than entropy
or mutual information as the basic principle. These divergences have
several advantages compared with entropy and mutual information.

To each convex optimization problem one can associate a Bregman divergence.
If the optimization problem is energy extraction in thermodynamics
the Bregman divergence is proportional to quantum relative entropy
that has some very desirable properties. These properties may be violated
if one looks at different optimization problems. Therefore one may
ask what is so special about energy extraction in thermodynamics,
but this important problem will not be covered in the present paper.
One advantage of studying divergence (and entropy) rather than conditional
mutual information is that divergence and its properties can be studied
for unipartite systems while conditional mutual information only makes
sense for multipartite systems. This is important because we do not
have a canonical way of forming product spaces in generalized probabilistic
theories. Bregman divergences with nice properties can be defined
on Jordan algebras and the existence of a nice Bregman divergence
rule out most other convex bodies as potential state spaces. Finally,
both entropy and conditional mutual information may be considered
as derived concepts based on divergence. This aspect will be the focus
of the present paper.

The paper is organized as follows. In Section \ref{sec:State-spaces}
we specify concepts like state space and measurement and we fix notation.
Jordan algebras and their most important properties are described
in Section \ref{sec:Jordan-algebras}. In Section \ref{sec:Entropy-in-Jordan-1}
it is proved that several different ways of defining entropy coincide
for Jordan algebras. Bregman divergences and their relation to optimization
are described in Section \ref{sec:Bregman-divergences}. Several conditions
related to the notion of sufficiency are defined. For Jordan algebras
these conditions are equivalent and the Bregman divergence is generated
by the entropy function. In Section \ref{sec:Information-causality}
we define conditional mutual information based on a Bregman divergence
and we demonstrate that the conditional mutual information has the
properties that are needed for information causality to be satisfied. 
We conclude with Section \ref{sec:conclusion} we summarize our results and state some open problems.

\section{State spaces\label{sec:State-spaces}}

Let $\mathcal{P}$ denote a set of \emph{preparations} of a physical
experiment. A \emph{mixed preparation} is a formal mixture $\sum s_{i}\cdot p_{i}$
where $p_{i}$ are preparations and $\left(s_{i}\right)_{i}$ is a
probability vector. The mixture $\sum s_{i}\cdot p_{i}$ is identified
with the preparation where $p_{i}$ is chosen with probability $s_{i}$.
A measurement $m$ maps each preparation in $\mathcal{P}$ into a
probability measure on the set of possible outcomes of the experiment.
We assume that $m$$\left(\sum s_{i}\cdot p_{i}\right)=\sum s_{i}\cdot m\left(p_{i}\right).$
Let $\mathcal{M}$ denote the set of measurements that can be performed
by an observer (or a group of observers). If $m\left(p_{1}\right)=m\left(p_{2}\right)$
for all measurements $m\in\mathcal{M}$ then we say that $p_{1}$
and $p_{2}$ represent the same \emph{state}. The set of states is
called the state space, and with this Bayesian definition of a state
the state space will depend on the set of feasible measurements. In
particular, the state spaces of two different observers may be different
because they may have different sets of measurements. A group of observers
may have a different state space than any of the individual observers
because the set of joint measurements may be larger than the set of
measurements that can be performed by any of the individual observers.

For simplicity we will assume that the state spaces are \emph{convex
bodies} $\Omega$, i.e. convex compact sets spanned by finitely many
elements. The extreme point are called \emph{pure states}. Any convex
body can be embedded in the pointed cone $\Omega_{+}$ consisting of
formal products $t\cdot\sigma$ where $\sigma$ is a state and $t$
is a positive real number called the \emph{trace} of $t\cdot\sigma$.
The notation is $\mathrm{tr}\left(t\cdot\sigma\right)=t.$ The elements
in the cone are called \emph{positive operators} or un-normalized
states. The cone is called the \emph{state cone}. Positive elements
can be added by 
\[
t_{1}\cdot\sigma_{1}+t_{2}\cdot\sigma_{2}=\left(t_{1}+t_{2}\right)\cdot\left(\frac{t_{1}}{t_{1}+t_{2}}\sigma_{1}+\frac{t_{2}}{t_{1}+t_{2}}\sigma_{2}\right).
\]
The state cone spans a partially ordered vector space $V_{\Omega}$
and the trace extends linearly to $V_{\Omega}$. Thus, the states
may be considered as positive elements of an ordered vector space
with trace 1.

Let $m\in\mathcal{M}$ denote a measurement with values $v$ in some
set $\mathcal{V}.$ If $\sigma$ is a state then the measurement is
given by a probability measure $m\left(\sigma\right)$ over $\mathcal{V}.$
Thus for each $v\in\mathcal{V}$ we have a probability $m\left(\sigma\right)\left(v\right)\in\left[0,1\right].$
For each $v$ the measurement $m$ maps $\Omega$ into $\left[0,1\right]$
and such a mapping is called a \emph{test} and it is an element in
$\Omega_{+}^{*}$ , i.e. the dual cone of the positive elements. In
the literature on generalized probabilistic theories a test is often
called an effect, but in this paper it is called a test, which is
the well established in the statistical literature. The test that
maps $x\in V_{\Omega}$ into $\lambda\mathrm{tr}\left(x\right)$ will
be denoted $\lambda.$ In particular the test $1$ maps $\Omega$
into $1.$ Since the total probability of a measurement is 1 we have
$\sum_{v}m\left(\cdot\right)\left(v\right)=1.$ A measurement can
be represented as a test valued measure. In the Hilbert space formalism
the tests are given by positive operators and the measurements are
given by positive operator valued measures (POVM). We say that two
states $\rho$ and $\sigma$ are mutually singular if there exists
a test $\phi$ such that $\phi\left(\rho\right)=0$ and $\phi\left(\sigma\right)=1.$

Let $m_{1},m_{2}\in\mathcal{M}$ with values in $\mathcal{V}_{1}$
and $\mathcal{V}_{2}.$ If $M:V_{1}\to V_{2}$ is some map such that
\[
m_{2}\left(\cdot\right)\left(v_{2}\right)=\sum_{v_{1}:M\left(v_{1}\right)=v_{2}}m_{1}\left(\cdot\right)\left(v_{1}\right)
\]
 then the measurement $m_{1}$ is at least as informative about the
state as $m_{2}$, and $m_{1}$ is called a \emph{fine-graining} of
$m_{2}.$ If 
\[
m_{2}\left(\cdot\right)\left(v_{2}\right)\propto m_{1}\left(\cdot\right)\left(v_{1}\right)
\]
for all values $v_{1}$ for which $M\left(v_{1}\right)=v_{2}$, then
the fine-graining is said to be \emph{trivial}. A measurement is \emph{fine-grained}
if all fine-grainings are trivial. Note that a measurement $m$ is
fine grained if all tests $m\left(\cdot\right)\left(v\right)$ lie
on extreme rays of $\Omega_{+}^{*}$. Therefore any measurement has
a fine-graining that is fine grained when the state space $\Omega$
is a convex body.

Let $\Omega_{1}$ and $\Omega_{2}$ denote two state spaces. An affine
map $\Phi:\Omega_{1}\to\Omega_{2}$ is called and \emph{affinity}.
Let $S:\Omega_{1}\to\Omega_{2}$ and $R:\Omega_{2}\to\Omega_{1}$
denote affinities. If $R\circ S=id_{\Omega_{1}}$then $S$ is called
a \emph{section} and $R$ is called a \emph{retraction}. A \emph{frame}
is a section $S:\Omega_{1}\to\Omega_{2}$ where $\Omega_{1}$ is a
simplex. 

Let $\Omega$ denote the state space of a group of observers. The
set of measurements $M_{A}$ of a single observer Alice is a subset
of the set of all measurements $M$ of the whole group of observers.
Therefore the there is a surjective affinity $\mathbb{E}_{A}:\Omega\to\Omega_{A}$.
Assume that Alice and Bob are observers that can perform measurements
independently. Further assume that the choice of measurement made
by Alice does not influence the outcome of a measurement made by Bob
and that a choice of measurement made by Bob does not influence the
outcome of a measurement made by Alice. This is called the \emph{no-signaling
condition}. If Alice performs the measurement $m_{A}$ and Bob performs
the measurement $m_{B}$, then the joint measurement is denoted $m_{A}\otimes m_{B}$.
Further assume that Alice and Bob can communicate. Then Alice and
Bob can perform any measurement of the form $\sum s_{i}\cdot m_{A}\otimes m_{B}$.
If Alice and Bob together can only perform measurements of the form
$\sum s_{i}\cdot m_{A}\otimes m_{B}$ their joint state space is a
subset of $V_{\Omega_{A}}\otimes V_{\Omega_{B}}.$ Assume further
that Alice and Bob can prepare states individually. If Alice prepares
the state $\sigma_{A}$ and Bob prepares the state $\sigma_{B}$ then
their joints state is $\sigma_{A}\otimes\sigma_{B}\in V_{\Omega_{A}}\otimes V_{\Omega_{B}}.$
The convex hull of $\left\{ \left.\sigma_{A}\otimes\sigma_{B}\right|\sigma_{A}\in\Omega_{A}\textrm{ and }\sigma_{B}\in\Omega_{B}\right\} $
is denoted $\Omega_{A}\otimes_{min}\Omega_{B}$ and the elements are
called separable states. We assume that $\Omega_{A}\otimes_{min}\Omega_{B}\subseteq\Omega.$

\section{Jordan algebras\label{sec:Jordan-algebras}}

Here we will recall some fact and concepts related to Jordan algebras.
A more detailed exposition can be found in \cite{McCrimmon2004,Baes2007}.
In the Hilbert space formalism of quantum physics the states are represented
as density matrices on a complex Hilbert space. Classical probability
distributions can be identified with density matrices that are diagonal.
In the set of self adjoint matrices one may define a product $\bullet$
by 
\[
A\bullet B=\frac{1}{2}\left(AB+BA\right)\,.
\]
 This product makes the set of Hermitean matrices into an algebra
over the real numbers and the product $\bullet$ satisfies
\begin{equation}
A\bullet\left(B\bullet\left(A\bullet A\right)\right)=\left(A\bullet B\right)\bullet\left(A\bullet A\right).\label{eq:ass}
\end{equation}
With this equation fulfilled it is possible to define $A^{n}=A\bullet A\bullet\dots\bullet A$
without specifying where the parenthesis have to be placed. Further
we have that 
\begin{equation}
\sum_{i}A_{i}^{2}=0\label{eq:formReal}
\end{equation}
if and only if $A_{i}=0$ for all $i.$ The dimension of the algebra
is defined as the dimension of the Jordan algebra as a real vector
space. A finite dimensional algebra over the real numbers with a product
$\bullet$ satisfying the properties (\ref{eq:ass}) and (\ref{eq:formReal})
is called an \emph{Euclidean Jordan algebra}. 

Elements in an Euclidean Jordan algebra of the form $A\bullet A$
are called positive elements and they form a pointed cone. Further,
an Euclidean Jordan algebra has a \emph{trace} $\mathrm{tr}$ that
maps positive elements into positive numbers and such 
\[
\mathrm{tr}\left(\left(A\bullet B\right)\bullet C\right)=\mathrm{tr}\left(A\bullet\left(B\bullet C\right)\right)\,.
\]
A \emph{state in a Jordan algebra} is a positive element of trace
1. The \emph{rank of a Jordan algebra} is the Caratheodory rank of
the state space of algebra. An Euclidean Jordan algebra has an inner
product defined by 
\[
\left\langle A,B\right\rangle =\mathrm{tr}\left(A\bullet B\right).
\]
With this inner product the positive cone becomes \emph{self dual}.

An element $E$ of a Jordan algebra is \emph{idempotent} if $E^{2}=E.$
Elements $A$ and $B$ are \emph{orthogonal} if $A\bullet B=0.$ With
these definitions any element $A$ has a spectral decomposition
\[
A=\sum\lambda_{i}E_{i}
\]
 where $E_{i}$ are orthogonal idempotent. If the spectral values
$\lambda_{i}$ are different, the decomposition is unique. Therefore
one can define
\[
f\left(A\right)=\sum f\left(\lambda_{i}\right)E_{i}\,.
\]

The associative Euclidean Jordan algebras correspond to classical
probability theory, where the state space is a simplex. Any Euclidean
Jordan algebra $\mathcal{J}$ can be written as a direct sum $\bigoplus\mathcal{J}_{i}$
of Jordan algebras where each of the Jordan algebras $\mathcal{J}_{i}$
is simple. The simple Euclidean Jordan algebras belong to one of the
the following five types.
\begin{itemize}
\item $M_{n}\left(\mathbb{R}\right)$ Real valued Hermitean $n\times n$
matrices.
\item $M_{n}\left(\mathbb{C}\right)$ Complex valued Hermitean $n\times n$
matrices.
\item $M_{n}\left(\mathbb{H}\right)$ Quaternionic valued Hermitean $n\times n$
matrices.
\item $M_{3}\left(\mathbb{O}\right)$ Octonionic valued Hermitean $3\times3$
matrices.
\item $Jspin\left(d\right)$ Spin factors where the state space has the
shape of a $d$-dimensional solid ball.
\end{itemize}
The Jordan algebra $M_{3}\left(\mathbb{O}\right)$ is called the \emph{exceptional
Jordan algebra} and Jordan algebras that does not contain such an
exceptional component are called \emph{special Jordan algebras}. All
special Jordan algebras appear as sections of $M_{n}\left(\mathbb{C}\right)$
for some value of $n$. In this sense all special Jordan algebras
have representations as physical systems. If a section of the set
of complex valued Hermitean matrices is required to be completely
positive then the section can be represented as a set of complex valued
Hermitean matrices.

It is an important question why exactly the complex valued Hermitean
matrices are so good in modeling quantum physics compared with the
other simple Jordan algebras. Actually Adler has attempted to model
quantum theory using quaternions \cite{Adler1995}, and there have
been a number of attempts to let the exceptional Jordan algebra play
an active role in modeling physics \cite{Guenaydin1973,Manogue2010}.
One important property that single out the complex valued Hermitean
matrices is that there is a canonical tensor product construction
within the category of complex valued Hermitean matrices with completely
positive maps as morphisms \cite{Barnum2016}. 
\begin{example}
\label{exa:reeltprodukt}Assume that the whole state space $\Omega$
can be represented as real non-negative definite $4\times4$ matrices
with trace 1. The dimension of this state space is 9. Let $A$ and
$B$ denote a $2\times2$ real Hermitean matrices. Then $A\otimes B$
can embedded in $\Omega$ as \label{sec:Entropy-in-Jordan}
\begin{multline*}
\left(\begin{array}{cc}
a_{11} & a_{12}\\
a_{21} & a_{22}
\end{array}\right)\otimes\left(\begin{array}{cc}
b_{11} & b_{12}\\
b_{21} & b_{22}
\end{array}\right)=\\
\left(\begin{array}{cccc}
a_{11}b_{11} & a_{11}b_{12} & a_{12}b_{11} & a_{12}b_{12}\\
a_{11}b_{21} & a_{11}b_{22} & a_{12}b_{21} & a_{12}b_{22}\\
a_{21}b_{11} & a_{21}b_{12} & a_{22}b_{11} & a_{22}b_{12}\\
a_{21}b_{21} & a_{21}b_{22} & a_{22}b_{21} & a_{22}b_{22}
\end{array}\right).
\end{multline*}
The vector space of Hermitean $2\times2$ matrices has dimension 3.
Therefore the tensor product has dimension 9. Hence the set of tensors
with trace 1 has dimension 8, so it has a lower dimension than set
of states on the whole space. Therefore there are joint states on
the whole space that cannot be distinguished by local measurements.
Hence the tomography condition is not fulfilled. 
\end{example}
There are a number of ways to characterize Jordan algebras. Above
we have defined the Jordan algebras algebraically. A classic result
is that a real vector space with a self-dual homogeneous cone can
be represented as a Jordan algebra \cite{Jordan1934}. A new result
is that a state space that is spectral and where any pair of frames
can be mapped into each other, can be represented by a Jordan algebra
\cite{Barnum2019}.

For Jordan algebras it is possible to define a well-behaved entropy
function and an associated divergence function. In \cite{Harremoes2017c}
it was proved that if a state space has rank 2 and it has a monotone
Bregman divergence then it can be represented as a Jordan algebra
(spin factor). Similar representation theorems for state spaces of
higher rank are not yet available, so in this paper we focus on other
consequences of the existence of entropy function or Bregman divergences.

\section{Entropy in Jordan algebras\label{sec:Entropy-in-Jordan-1}}

In generalized probabilistic theories there are two ways of defining
entropy \cite{Short2010}. The \emph{decomposition entropy} of a state
$\sigma$ is given by 
\[
\breve{H}\left(\sigma\right)=\inf_{\sum p_{i}\cdot\sigma_{i}=\sigma}H\left(\left(p_{i}\right)_{i}\right).
\]
Here the infimum is taken over all mixtures $\sum p_{i}\cdot\sigma_{i}=\sigma$
where $\sigma_{i}$ are pure states and $H\left(\left(p_{i}\right)_{i}\right)$
denotes the Shannon entropy of the probability vector $\left(p_{i}\right)_{i}$.
Versions of this definition can also be found in \cite{Harremoes2017b},
but they dates back to \cite{Ullman1989}. Note that the definition
of spectral entropy in \cite{Krumm2017} is closely related but slightly
different.

Following \cite{Short2010} one can define the \emph{fine grained
entropy} of a state in a generalized probabilistic theory by
\[
\hat{H}\left(\sigma\right)=\inf_{m}H\left(m\left(\sigma\right)\right)
\]
where the infimum has been taken over all fine grained measurements $m$
on $\Omega$. This fine grained entropy is a strictly concave function.
\begin{lemma}
\label{lem:spectralentropy}If the state space $\Omega$ is spectral
a decomposition that minimizes the decomposition entropy is spectral.
\end{lemma}
\begin{svmultproof}
This was essentially proved in \cite{Harremoes2017b} although the
terminology regarding spectrality was slightly different. 
\end{svmultproof}

\begin{theorem}
\label{lem:entropiulighed}If the state space $\Omega$ is spectral
then for any state $\sigma$ the following inequality holds
\[
\hat{H}\left(\sigma\right)\leq\breve{H}\left(\sigma\right).
\]
\end{theorem}
\begin{svmultproof}
Let $\sigma=\sum p_{i}\sigma_{i}$ be a decomposition of $\sigma$
where the states $\sigma_{i}$ are pure. To this decomposition there
corresponds a measurement $m$ such that 
\[
m\left(\sigma\right)\left(i\right)=p_{i}\,.
\]
Since this measurement is fine grained we have 
\[
\hat{H}\left(\sigma\right)  \leq H\left(m\left(\sigma\right)\right)
  =H\left(\left(p_{i}\right)_{i}\right)\,.
\]
Therefore 
\[
\hat{H}\left(\sigma\right)  \leq\inf_{\sum p_{i}\sigma_{i}=\sigma}H\left(\left(p_{i}\right)_{i}\right)
  =\breve{H}\left(\sigma\right).
\]
\end{svmultproof}

\begin{theorem}
If the state space $\Omega$ is spectral and the cone $\Omega_{+}$
is self dual then 
\begin{equation}
\hat{H}\left(\sigma\right)=\breve{H}\left(\sigma\right)=\textrm{-}\left\langle \sigma,\ln\left(\sigma\right)\right\rangle .\label{eq:EntropiLigning}
\end{equation}
\end{theorem}
\begin{proof}
Let $M$ denote a fine grained measurement. The measurement is given
by a positive test valued measure, i.e. there exists $\rho_{j}\geq0$
such that $\sum\rho_{j}=1$ and such that
\[
M\left(\rho\right)\left(j\right)=\left\langle \rho_{j},\rho\right\rangle \,.
\]
Since the measurement is fine grained $\rho_{j}$ must be states.
Thus,
\[
M\left(\sigma\right)\left(j\right)=\left\langle \rho_{j},\sigma\right\rangle =\left\langle \rho_{j},\sum_{i}p_{i}\sigma_{i}\right\rangle =\sum_{i}p_{i}\left\langle \rho_{j},\sigma_{i}\right\rangle.
\]
If $\text{\ensuremath{\tilde{\sigma}} is the state \ensuremath{\sum_{i}\frac{1}{r}}\ensuremath{\cdot}\ensuremath{\sigma_{i}}}$
then 
\[
M\left(\tilde{\sigma}\right)=\left\langle \rho_{j},\sum_{i}\frac{1}{r}\cdot\sigma_{i}\right\rangle _{j}=\frac{1}{r}\left\langle \rho_{j},1\right\rangle _{j}=\frac{1}{r}.
\]
the Markov kernel $\left(p_{i}\right)_{i}\to\sum_{i}p_{i}\left\langle \rho_{j},\sigma_{i}\right\rangle _{j}$
maps the uniform distribution $\left(\frac{1}{r}\right)_{i}$ into
the uniform distribution $\left(\frac{1}{r}\right)_{i}$, i.e. the
Markov kernel is bi-stochastic. Since bi-stochastic Markov kernels
increase entropy we have 
\begin{align*}
H\left(M\left(\sigma\right)\right) & =H\left(\left\langle \rho_{j},\sigma\right\rangle _{j}\right)\geq H\left(\left(p_{i}\right)_{i}\right)=\textrm{-}\left\langle \sigma,\ln\left(\sigma\right)\right\rangle .
\end{align*}
Therefore 
\begin{equation}
\textrm{-}\left\langle \sigma,\ln\left(\sigma\right)\right\rangle \leq\hat{H}\left(\sigma\right).\label{eq:fineEntUlighed-1}
\end{equation}
Now the result is obtained by combining Lemma \ref{lem:entropiulighed}
and Theorem \ref{lem:spectralentropy} with inequality (\ref{eq:fineEntUlighed-1}).
\end{proof}

\begin{definition}
The entropy $H$ of a state $\sigma$ in a Jordan algebra is given
as the common value of any of the expressions given in Equation (\ref{eq:EntropiLigning}).
\end{definition}
\begin{corollary}
In a finite Euclidean Jordan algebra the entropy $\textrm{-}\left\langle \sigma,\ln\left(\sigma\right)\right\rangle $
is a concave function.
\end{corollary}
\begin{svmultproof}
Concavity of $H$ follows because $H$ equals the fine grained entropy
and the fine grained entropy is concave \cite{Short2010}. 
\end{svmultproof}

Concavity of the entropy function $H$ on Jordan algebras was proved
in \cite{Harremoes2017b} with a more involved proof.

\section{Bregman divergences and sufficiency conditions\label{sec:Bregman-divergences}}

We consider a optimization problem where we want to optimize some
quantity defined on the state space. In thermodynamics the goal is
typically to extract energy from the system by some feasible interaction
with the system. Our approach makes sense for any convex optimization
problem and in principle the function may represent other objectives
such as the amount of money one may obtained by trading or the code
length that is obtained after using a certain data compression procedure.
Various examples of such optimization problems are given in \cite{Harremoes2017}.
In this paper the objective function will be energy. 

Assume that the system is in state $\rho\in\Omega$ and that we apply
some action $a$ from a set of feasible actions $\mathcal{A}$. Then
the mean energy that we extract will be denoted 
\[
\left\langle a,\rho\right\rangle 
\]
and it is an affine function of the state $\rho.$ An action $a$
will be identified with this function $\rho\to\left\langle a,\rho\right\rangle $
so that the actions are considered as elements in the dual space of
the state space. We can define the free energy of state $\rho$ as
\[
F\left(\rho\right)=\sup_{a\in\mathcal{A}}\left\langle a,\rho\right\rangle \,.
\]
In thermodynamics Helmholz free energy is given as $F=U-TS$ so that
the free energy is an affine function minus a term that is proportional
to the entropy function. Then $F$ is a convex function of $\rho.$
The regret of doing action $a$ if the state is $\rho$ is defined
as 
\[
D_{F}\left(\rho,a\right)=F\left(\rho\right)-\left\langle a,\rho\right\rangle .
\]
The interpretation of the regret function is as follows. Assume that
the system is in state $\rho$ but one uses a sub-optimal action $a.$
Then the regret measures the difference between the energy that one
could have extracted $F\left(\rho\right)$ and the energy that one
extracts using action $a.$ For simplicity we will assume that $F$
is differentiable so that to each state $\rho$ there exists a unique
action $a_{\rho}$ such that $F\left(\rho\right)=\left\langle a,\rho\right\rangle .$
For states $\rho,\sigma\in\Omega$ the Bregman divergence is defined
as 
\[
D_{F}\left(\rho,\sigma\right)=D_{F}\left(\rho,a_{\sigma}\right).
\]
 It measures the regret of acting as if the state were $\sigma$ if
it actually is $\rho.$ The Bregman divergence is given by 
\begin{multline*}
D_{F}\left(\rho,\sigma\right)=\\
F\left(\rho\right)-\left(F\left(\sigma\right)+\frac{\mathrm{d}}{\mathrm{d}t}F\left(\left(1-t\right)\sigma+t\rho\right)_{\mid t=0}\right).
\end{multline*}
The formula for the Bregman divergence is often written in terms of
the gradient.
\[
D_{F}\left(\rho,\sigma\right)=F\left(\rho\right)-\left(F\left(\sigma\right)+\left\langle \left.\nabla F\left(\sigma\right)\right|\rho-\sigma\right\rangle \right).
\]

\begin{proposition}[{\cite[Lemma 17]{Harremoes2017b}}]
For Hermitean matrices $A$ and $B$ we have 
\[
\frac{\mathrm{d}}{\mathrm{d}t}\left(\mathrm{tr}\left(f\left(A+tB\right)\right)\right)_{\mid t=0}=\left\langle f'\left(A\right),B\right\rangle \,.
\]
\end{proposition}
\begin{example}
Assume that the state space can be represented as a state space of
a Jordan algebra. Let $F\left(\sigma\right)=\left\langle \sigma,\ln\left(\sigma\right)\right\rangle $
denote the negative of the entropy. The Bregman divergence corresponding
to $F$ can be computed as 

\begin{dmath}
D_{F}\left(\rho,\sigma\right)=F\left(\rho\right)-\left\{ F\left(\sigma\right)+\frac{\mathrm{d}}{\mathrm{d}t}F\left(\left(1-t\right)\sigma+t\rho\right)_{\mid t=0}\right\} \\
=\left\langle \rho,\ln\left(\rho\right)\right\rangle \\
-\left\{ \left\langle \sigma,\ln\left(\sigma\right)\right\rangle +\left\langle \ln\left(\sigma\right)+1,\rho-\sigma\right\rangle \right\} \\
=\left\langle \rho,\ln\left(\rho\right)-\ln\left(\sigma\right)\right\rangle -\mathrm{tr}\left(\rho-\sigma\right).\label{eq:KLformel}
\end{dmath}

We call this quantity the \emph{information divergence} and denote
it as $D\left(\rho\Vert\sigma\right)$. Note that the last term vanish
if $\rho$ and $\sigma$ are states. If the Jordan algebra is associative
we get Kulback-Leibler divergence given by 
\[
D\left(P\Vert Q\right)=\sum p_{i}\ln\frac{p_{i}}{q_{i}}\,.
\]
If the Jordan algebra is a $C^{*}$-algebra $F$ is minus the von
Neumann entropy the information divergence equals quantum information
divergence (quantum relative entropy) given by
\[
D\left(\rho\Vert\sigma\right)=\mathrm{tr}\left(\rho\left(\ln\rho-\ln\sigma\right)\right)\,.
\]
\end{example}
There are a number of conditions that some regret functions and Bregman
divergences may have.
\begin{definition}
The Bregman divergence $D_{F}$ is \emph{monotone} if $D_{F}\left(\Phi\left(\rho\right),\Phi\left(\sigma\right)\right)\leq D_{F}\left(\rho,\sigma\right)$
for any affinity $\Phi:\Omega\to\Omega$.
\end{definition}
We note that monotonicity is associated with the decrease of free
energy for a closed thermodynamic system. It is possible to define
the regret $D_{F}\left(\rho,\sigma\right)$ even if the function $F$
is not differentiable, but if such a regret function is monotone then
$F$ is automatically differentiable \cite{Harremoes2017}. In the
rest of this paper we shall focus entirely on the case when $F$ is
differentiable and the regret between states is given by the Bregman
divergence.
\begin{theorem}
Information divergence is monotone on special Jordan algebras.
\end{theorem}
\begin{svmultproof}
Let $\Omega$ denote the state space of a special Jordan algebra.
Then there exists a section $S:\Omega\to M_{n}\left(\mathbb{C}\right)_{+}^{1}$
with a corresponding retraction $R:M_{n}\left(\mathbb{C}\right)_{+}^{1}\to\Omega.$
Let $\Phi:\Omega\to\Omega$ denote some affinity. Then $S\circ\Phi\circ R$
is an affinity $M_{n}\left(\mathbb{C}\right)_{+}^{1}\to M_{n}\left(\mathbb{C}\right)_{+}^{1}$.
Then 
\begin{multline*}
D\left(\left.\Phi\left(\rho\right)\right\Vert \Phi\left(\sigma\right)\right)
=D\left(\left.S\left(\Phi\left(\rho\right)\right)\right\Vert S\left(\Phi\left(\sigma\right)\right)\right)\\
=D\left(\left.\left(S\circ\Phi\circ R\right)\left(S\left(\rho\right)\right)\right\Vert \left(S\circ\Phi\circ R\right)\left(S\left(\sigma\right)\right)\right)\\
\leq D\left(\left.S\left(\rho\right)\right\Vert S\left(\sigma\right)\right)
=D\left(\left.\rho\right\Vert \sigma\right)\,.
\end{multline*}
Here we have used that information divergence is monotone on $M_{n}\left(\mathbb{C}\right)_{+}^{1}$
\cite{Mueller-Hermes2017}.
\end{svmultproof}

It is not known if information divergence is monotone on the exceptional
Jordan algebra. Let $\rho_{\theta}$ denote a family of states and
let $\Phi$ denote an affinity $\Phi:\Omega\to\Omega$. Then $\Phi$
is said to be \emph{sufficient} for $\rho_{\theta}$ if there exists
a \emph{recovery map} $\Psi:\Omega\to\Omega$, i.e. an affinity such
that $\Psi\left(\Phi\left(\rho_{\theta}\right)\right)=\rho_{\theta}$. 
\begin{definition}
A Bregman divergence $D_{F}$ is said to satisfy \emph{sufficiency}
if $D_{F}\left(\Phi\left(\rho\right),\Phi\left(\sigma\right)\right)=D_{F}\left(\rho,\sigma\right)$
whenever $\Phi$ is sufficient for $\rho,\sigma.$
\end{definition}
It is easy to prove that monotonicity implies sufficiency. Further
it is easy to prove that sufficiency implies the property called statistical
locality as defined below.
\begin{definition}
A Bregman divergence $D_{F}$ satisfies \emph{statistical locality}
if $\rho\bot\sigma_{i}$ implies 
\[
D_{F}\left(\rho,\left(1-t\right)\cdot\rho+t\cdot\sigma_{1}\right)=D_{F}\left(\rho,\left(1-t\right)\cdot\rho+t\cdot\sigma_{2}\right).
\]
\end{definition}
\begin{proposition}
In an Euclidean Jordan algebra Information divergence satisfies statistical locality.
\end{proposition}
\begin{svmultproof}
Assume that $\rho,\,\sigma_{1},$ and $\sigma_{2}$ are states and
that $\rho\bot\sigma_{i}$. Then 
\begin{dmath*}
D\left(\rho\left\Vert \left(1-t\right)\cdot\rho+t\cdot\sigma_{1}\right.\right)
=\left\langle \rho,\ln\left(\rho\right)-\ln\left(\left(1-t\right)\cdot\rho+t\cdot\sigma_{1}\right)\right\rangle 
=\left\langle \rho,\ln\left(\rho\right)-\ln\left(\left(1-t\right)\cdot\rho\right)\right\rangle =\textrm{-}\ln\left(1-t\right).
\end{dmath*}
\end{svmultproof}

\begin{theorem}
\label{thm:local}If the state space $\Omega$ can be represented
as the state space of a Jordan algebra of rank at least 3 then a statistically
local Bregman divergence $D_{F}$ is proportional to information divergence
given by Equation (\ref{eq:KLformel}). There exists a constant $c>0$
such that the function $F$ equals $c\cdot\left\langle \rho,\ln\rho\right\rangle $
plus an affine function on $\Omega$. 
\end{theorem}
\begin{svmultproof}
The theorem was proved for finite $C^{*}$-algebras in \cite{Harremoes2017},
but the proof is the same for more general Jordan algebras. 
\end{svmultproof}

The theorem implies under certain conditions the following conditions
are equivalent
\begin{itemize}
\item Monotonicity, 
\item Sufficiency
\item Statistical locality
\item The Bregman divergence is proportional to information divergence.
\item The objective function $F$ is proportional to entropy plus an affine
function.
\end{itemize}
If the state space has rank 2 these conditions are not equivalent
and this special case was studied in great detail in \cite{Harremoes2017c}. 

\section{Information causality\label{sec:Information-causality}}

Consider a bipartite system with Alice and Bob as observers. We assume
the no-signaling condition and local tomography are fulfilled so that
a joint state can be described as an element in the tensor product
of local vector spaces. Let $U_{A}$ and $U_{b}$ denote order units
of Alice and Bob.

Let $F$ denote some payoff function on a joint system with regret
function $D_{F}.$ We will assume that the regret function $D_{F}$
satisfies monotonicity. Then $F$ is differentiable and $D_{F}$ is
a Bregman divergence. Therefore $D_{F}$ is given by 
\[
D_{F}\left(\rho,\sigma\right)=F\left(\rho\right)-\left(F\left(\sigma\right)+\left\langle \left.\nabla F\left(\sigma\right)\right|\rho-\sigma\right\rangle \right)\,.
\]
The following proposition is well-known if the affine
combination is a convex combination.
\begin{proposition}
If $\sum_{i}t_{i}=1$ and the affine combination $\bar{\rho}=\sum_{i}t_{i}\cdot\rho_{i}$
is a state then the Bregman identity holds:
\begin{equation}
\sum_{i}t_{i}\cdot D_{F}\left(\rho_{i},\sigma\right)=\sum_{i}t_{i}\cdot D_{F}\left(\rho_{i},\bar{\rho}\right)+D_{F}\left(\bar{\rho},\sigma\right).\label{eq:BregmanIdentity}
\end{equation}
 
\end{proposition}
\begin{svmultproof}
We expand the right hand side of (\ref{eq:BregmanIdentity}) and get
\begin{dmath*}
\sum_{i}t_{i}\cdot D_{F}\left(\rho_{i},\bar{\rho}\right)+D_{F}\left(\bar{\rho},\sigma\right)\\
=\sum_{i}t_{i}\cdot\left(F\left(\rho_{i}\right)-\left(F\left(\bar{\rho}\right)+\left\langle \left.\nabla F\left(\bar{\rho}\right)\right|\rho_{i}-\bar{\rho}\right\rangle \right)\right)\\
+F\left(\bar{\rho}\right)-\left(F\left(\sigma\right)+\left\langle \left.\nabla F\left(\sigma\right)\right|\bar{\rho}-\sigma\right\rangle \right)\,.
\end{dmath*}
We can re-arrange the terms and use that
\[
\bar{\rho}=\sum_{i}t_{i}\cdot\rho_{i}
\]
 to get
\begin{dmath*}
\sum_{i}t_{i}\cdot F\left(\rho_{i}\right)
-\left(\sum_{i}t_{i}\cdot F\left(\bar{\rho}\right)+\left\langle \left.\nabla F\left(\bar{\rho}\right)\right|\sum_{i}t_{i}\cdot\rho_{i}-\bar{\rho}\right\rangle \right)\\
+F\left(\bar{\rho}\right)-\left(F\left(\sigma\right)+\left\langle \left.\nabla F\left(\sigma\right)\right|\bar{\rho}-\sigma\right\rangle \right)\\
=\sum_{i}t_{i}\cdot F\left(\rho_{i}\right)-\left(F\left(\bar{\rho}\right)+\left\langle \left.\nabla F\left(\bar{\rho}\right)\right|\bar{\rho}-\bar{\rho}\right\rangle \right)\\
+F\left(\bar{\rho}\right)-\left(F\left(\sigma\right)+\left\langle \left.\nabla F\left(\sigma\right)\right|\bar{\rho}-\sigma\right\rangle \right)\,.
\end{dmath*}
Therefore the right hand side of Equation (\ref{eq:BregmanIdentity})
reduces to
\begin{dmath*}
\sum_{i}t_{i}\cdot F\left(\rho_{i}\right)-\left(F\left(\sigma\right)+\left\langle \left.\nabla F\left(\sigma\right)\right|\bar{\rho}-\sigma\right\rangle \right)\\
=\sum_{i}t_{i}\cdot\left(F\left(\rho_{i}\right)-\left(F\left(\sigma\right)+\left\langle \left.\nabla F\left(\sigma\right)\right|\rho_{i}-\sigma\right\rangle \right)\right)\\
=\sum_{i}t_{i}\cdot D_{F}\left(\rho_{i},\sigma\right),
\end{dmath*}
which is the left hand side of Equation (\ref{eq:BregmanIdentity})
and this completes the proof.
\end{svmultproof}

\begin{theorem}
Assume that $\Omega\subset V_{A}\otimes V_{B}$. If $\rho_{1},\rho_{2}\in\Omega_{A}$
and $\sigma_{1},\sigma_{2}\in\Omega_{B}$ and $D_{F}$ satisfies sufficiency
then
\[
D_{F}\left(\rho_{1}\otimes\sigma_{1},\rho_{2}\otimes\sigma_{1}\right)=D_{F}\left(\rho_{2}\otimes\sigma_{2},\rho_{2}\otimes\sigma_{2}\right)\,.
\]
\end{theorem}
\begin{svmultproof}
To see this define
\begin{align*}
\Phi\left(\pi\right) & =\mathbb{E}_{A}\left(\pi\right)\otimes\sigma_{1}\,,\\
\Psi\left(\pi\right) & =\mathbb{E}_{A}\left(\pi\right)\otimes\sigma_{2}\,.
\end{align*}
Then
\begin{align*}
\Phi\left(\rho_{i}\otimes\sigma_{2}\right) & =\rho_{i}\otimes\sigma_{1}\,,\\
\Psi\left(\rho_{i}\otimes\sigma_{1}\right) & =\rho_{i}\otimes\sigma_{2}\,.
\end{align*}
The result is obtained by sufficiency of $D_{F}$. 
\end{svmultproof}

If $\rho_{1},\rho_{2}\in\Omega_{A}$ we may write $D_{F}\left(\rho_{1},\rho_{2}\right)$
as an abbreviation for $D_{F}\left(\rho_{2}\otimes\sigma,\rho_{2}\otimes\sigma\right)$
where some arbitrary state $\sigma\in\Omega_{B}$ is used. 
\begin{definition}
Let $\sigma$ denote a state on a system with a bipartite subsystem
composed of subsystems labeled $A$ and $B$. Then the mutual information
between the subsystem $A$ and subsystem $B$ is defined as
\begin{equation}
I_{\sigma}\left(A;B\right)=D_{F}\left(\sigma_{AB},\sigma_{A}\otimes\sigma_{B}\right).\label{eq:mutual}
\end{equation}
\end{definition}
\begin{theorem}
If the Bregman divergence $D_{F}$ is monotone then mutual information
satisfies the following two conditions.

\textbf{Consistency} If the system has a bipartite subsystem consisting
of two classical subsystems $A$ and $B$ then the mutual information
restricted to the bipartite subsystem is proportional to classical
mutual information.

\textbf{Data processing inequality} If $\Phi:V_{B}\to V_{B}$
is a positive trace conserving affinity then
\[
I_{\sigma}\left(A;B\right)\geq I_{\left(id\otimes\Phi\right)\left(\sigma\right)}\left(A;B\right).
\]
\end{theorem}
\begin{svmultproof}
\textbf{Consistency} If the subsystems defined by Alice and Bob are
classical and non-trivial then the rank of their joint state space
is at least $2\times2=4.$ When the rank of the state space is least
3 the function $F$ is a linear function of the Shannon entropy and
therefore the mutual information defined by (\ref{eq:mutual}) is
proportional to the classical mutual information.

\textbf{Data processing inequality} Assume that 
\[
\Phi:V_{B}\to V_{B}
\]
is a positive trace conserving affinity. Then $\tilde{\Phi}=id\otimes\Phi$
is given by $\tilde{\Phi}\left(\sigma_{A}\otimes\sigma_{B}\right)=\sigma\otimes\Phi\left(\sigma_{B}\right)$
and 
\begin{align*}
I_{\sigma}\left(A;B\right) & =D_{F}\left(\sigma_{AB},\sigma_{A}\otimes\sigma_{B}\right)\\
 & \leq D_{F}\left(\tilde{\Phi}\left(\sigma_{AB}\right),\tilde{\Phi}\left(\sigma_{A}\otimes\sigma_{B}\right)\right)\\
 & =D_{F}\left(\tilde{\Phi}\left(\sigma_{AB}\right),\sigma_{A}\otimes\Phi\left(\sigma_{B}\right)\right)\\
 & =I_{\tilde{\Phi}\left(\sigma\right)}\left(A;B\right)\,,
\end{align*}
which completes the proof.
\end{svmultproof}

In probability theory one may define entropy as self information via
\[
H\left(A\right)=I\left(A,A\right).
\]
This is not possible in quantum theory because the different sub-spaces
in a tensor product decomposition have to be distinct. In probability
theory this is not a problem and cloning is allowed i.e. one is allowed
to form identical copies a state. In probability theory one gets
\begin{align*}
H\left(AB\right) & =I\left(AB,AB\right)\\
 & =I\left(A,AB\right)+I\left(B,AB\mid A\right)\\
 & \geq I\left(A,AB\right)\\
 & =I\left(A,A\right)+I\left(A,B\mid A\right)\\
 & \geq I\left(A,A\right)\\
 & =H\left(A\right).
\end{align*}
Therefore in probability theory the entropy of a subsystem is less
than the entropy of the full system.
\begin{definition}
A Bregman divergence $D_{F}$ on a bipartite system is \emph{additive}
if 
\[
D_{F}\left(\rho_{A}\otimes\rho_{B},\sigma_{A}\otimes\sigma_{B}\right)=D_{F}\left(\rho_{A},\sigma_{A}\right)+D_{F}\left(\rho_{B},\sigma_{B}\right).
\]
\end{definition}
\begin{theorem}
If the state spaces $\Omega_{A}$ and $\Omega_{B}$ can be represented
as state spaces of Jordan algebras $\mathcal{J}_{A}$ and $\mathcal{J}_{B}$,
and if $D_{F}$ satisfies sufficiency then $D_{F}$ is additive.
\end{theorem}
\begin{svmultproof}
Let $c_{A}$ and $c_{B}$ denote distributions that maximize the fine
grained entropy distributions in each of the algebras. Then $D_{F}$
equals $D_{\tilde{F}}$ where 
\[
\tilde{F}\left(\sigma\right)=D_{F}\left(\sigma,c_{A}\otimes c_{B}\right).
\]
Let $\rho_{A}$ and $\rho_{B}$ denote states in the state spaces
$\Omega_{A}$ and $\Omega_{B}.$ Then $\rho_{A}$ and $\rho_{B}$
generate associative sub-algebras $\mathcal{A}_{A}\subseteq\mathcal{J}_{A}$
and $\mathcal{A}_{B}\subseteq\mathcal{J}_{B}$ with classical state
spaces. Now the restriction of $D_{F}$ to $\mathcal{A}_{A}\otimes\mathcal{A}_{B}$
satisfies sufficiency and according to Theorem \ref{thm:local} $D_{F}$
is proportional to information divergence. Therefore 
\[
D_{F}\left(\rho_{A}\otimes\rho_{B},c_{A}\otimes c_{B}\right)=D_{F}\left(\rho_{A},c_{A}\right)+D_{F}\left(\rho_{B},c_{B}\right)
\]
because information divergence is additive on classical state spaces.
Define 
\begin{align*}
\tilde{F}_{A}\left(\rho_{A}\right) & =D_{F}\left(\rho_{A},c_{A}\right)\,,\\
\tilde{F}_{B}\left(\rho_{B}\right) & =D_{F}\left(\rho_{B},c_{B}\right)\,.
\end{align*}
With this notation
\[
\tilde{F}\left(\rho_{A}\otimes\rho_{B}\right)=\tilde{F}_{A}\left(\rho_{A}\right)+\tilde{F}_{B}\left(\rho_{B}\right).
\]
Thus 
\begin{multline*}
D_{F}\left(\rho_{A}\otimes\rho_{B},\sigma_{A}\otimes\sigma_{B}\right)=\tilde{F}\left(\rho_{A}\otimes\rho_{B}\right)\\
\,\,\,\,-\left(\begin{array}{c}
\tilde{F}\left(\sigma_{A}\otimes\sigma_{B}\right)+\\
\left\langle \left.\nabla\tilde{F}\left(\sigma_{A}\otimes\sigma_{B}\right)\right|\rho_{A}\otimes\rho_{B}-\sigma_{A}\otimes\sigma_{B}\right\rangle 
\end{array}\right)\\
=\tilde{F}_{A}\left(\rho_{A}\right)+\tilde{F}_{B}\left(\rho_{B}\right)\\
\,\,\,\,-\left(\begin{array}{c}
\tilde{F}_{A}\left(\sigma_{A}\right)+\tilde{F}_{B}\left(\sigma_{B}\right)+\\
\left\langle \left.\nabla\tilde{F}_{A}\left(\sigma_{A}\right)+\nabla\tilde{F}_{B}\left(\sigma_{B}\right)\right|\rho_{A}\otimes\rho_{B}-\sigma_{A}\otimes\sigma_{B}\right\rangle 
\end{array}\right)\\
=\tilde{F}_{A}\left(\rho_{A}\right)-\left(\begin{array}{c}
\tilde{F}\left(\sigma_{A}\right)+\\
\left\langle \left.\nabla\tilde{F}_{A}\left(\sigma_{A}\right)\right|\rho_{A}\otimes\rho_{B}-\sigma_{A}\otimes\sigma_{B}\right\rangle 
\end{array}\right)\\
\,\,\,\,+\tilde{F}_{B}\left(\rho_{B}\right)-\left(\begin{array}{c}
\tilde{F}_{B}\left(\sigma_{B}\right)+\\
\left\langle \left.\nabla\tilde{F}_{B}\left(\sigma_{B}\right)\right|\rho_{A}\otimes\rho_{B}-\sigma_{A}\otimes\sigma_{B}\right\rangle 
\end{array}\right)\\
=D_{F}\left(\rho_{A},\sigma_{A}\right)+D_{F}\left(\rho_{B},\sigma_{B}\right)\,.
\end{multline*}
\end{svmultproof}

\begin{example}
If tensor products of $2\times2$ Hermitean matrices are embedded
in Hermitean $4\times4$ matrices as in Example \ref{exa:reeltprodukt}
then mutual information is additive.
\end{example}
\begin{lemma}
An additive monotone Bregman divergence satisfies the following identity
\begin{dmath}
D_{F}\left(\sigma_{AB},\rho_{A}\otimes\rho_{B}\right)=  D_{F}\left(\sigma_{AB},\sigma_{A}\otimes\rho_{B}\right) 
  +D_{F}\left(\sigma_{A},\rho_{A}\right).\label{eq:marginal}
\end{dmath}
\end{lemma}
\begin{svmultproof}
Any state $\sigma_{AB}$ can be written as an affine combination of
tensor products $\sigma_{AB}=\sum t_{i}\cdot\pi_{A,i}\otimes\pi_{B,i}.$
Then
\begin{dmath*}
D_{F}\left(\sigma_{AB},\rho_{A}\otimes\rho_{B}\right)=\sum t_{i}\cdot D_{F}\left(\pi_{A,i}\otimes\pi_{B,i},\rho_{A}\otimes\rho_{B}\right)\\
-\sum t_{i}\cdot D_{F}\left(\pi_{A,i}\otimes\pi_{B,i},\sigma_{AB}\right)\,.
\end{dmath*}
Using additivity it can be rewritten as 
\begin{multline*}
D_{F}\left(\sigma_{AB},\rho_{A}\otimes\rho_{B}\right)\\
=\sum t_{i}\cdot\left(D_{F}\left(\pi_{A,i},\rho_{A}\right)+D_{F}\left(\pi_{B,i},\rho_{B}\right)\right)\\
\,\,\,\,-\sum t_{i}\cdot D_{F}\left(\pi_{A,i}\otimes\pi_{B,i},\sigma_{AB}\right)\\
=\sum t_{i}\cdot D_{F}\left(\pi_{A,i},\rho_{A}\right)+\sum t_{i}\cdot D_{F}\left(\pi_{B,i},\rho_{B}\right)\\
-\sum t_{i}\cdot D_{F}\left(\pi_{A,i}\otimes\pi_{B,i},\sigma_{AB}\right)\,.
\end{multline*}
The Bregman identity (\ref{eq:BregmanIdentity}) gives
\begin{dmath*}
D_{F}\left(\sigma_{AB},\rho_{A}\otimes\rho_{B}\right)=\\
\sum t_{i}\cdot D_{F}\left(\pi_{A,i},\sigma_{A}\right)+D_{F}\left(\sigma_{A},\rho_{A}\right)\\
+\sum t_{i}\cdot D_{F}\left(\pi_{B,i},\rho_{B}\right)
-\sum t_{i}\hiderel{\cdot} D_{F}\left(\pi_{A,i}\otimes\pi_{B,i},\sigma_{AB}\right)\,.
\end{dmath*}
This can be re-arranged as
\begin{dmath*}
D_{F}\left(\sigma_{AB},\rho_{A}\otimes\rho_{B}\right)=
\sum t_{i}\hiderel{\cdot}\left(D_{F}\left(\pi_{A,i},\sigma_{A}\right)\hiderel{+}D_{F}\left(\pi_{B,i},\rho_{B}\right)\right)\\
-\sum t_{i}\cdot D_{F}\left(\pi_{A,i}\otimes\pi_{B,i},\sigma_{AB}\right)+D_{F}\left(\sigma_{A},\rho_{A}\right).
\end{dmath*}
Now additivity leads to 
\begin{multline*}
D_{F}\left(\sigma_{AB},\rho_{A}\otimes\rho_{B}\right)=\\
\sum t_{i}\cdot D_{F}\left(\pi_{A,i}\otimes\pi_{B,i},\sigma_{A}\otimes\rho_{B}\right)\\
-\sum t_{i}\cdot D_{F}\left(\pi_{A,i}\otimes\pi_{B,i},\sigma_{AB}\right)+D_{F}\left(\sigma_{A},\rho_{A}\right)\\
=D_{F}\left(\sigma_{AB},\sigma_{A}\otimes\rho_{B}\right)+D_{F}\left(\sigma_{A},\rho_{A}\right).
\end{multline*}
\end{svmultproof}

\begin{definition}
We define the \emph{conditional mutual information} on a tripartite
system as 
\end{definition}
\begin{multline*}
I_{\sigma}\left(A;B\mid C\right)=D_{F}\left(\sigma_{ABC},\sigma_{A}\otimes\sigma_{B}\otimes\sigma_{C}\right)\\
-D_{F}\left(\sigma_{AC},\sigma_{A}\otimes\sigma_{C}\right)-D_{F}\left(\sigma_{BC},\sigma_{B}\otimes\sigma_{C}\right).
\end{multline*}

In our definition of conditional mutual information the subsystems
$A,\,B,$ and $C$ should be distinct so that the tensor products
are defined. If the state space is a simplex, i.e. the system is classical,
then one may allow the subsystems to overlap. 
\begin{definition}
A function $I_{\sigma}$ on a multipartite system is called a \emph{separoid
function}\cite{Dawid2001,Harremoes2018a} if it satisfies the following three properties:
\begin{equation}
\begin{array}{ll}
\textrm{\textbf{Positivity}} & I_{\sigma}\left(A;B\mid C\right)\geq0\,.\\
\textrm{\textbf{Symmetry}} & I_{\sigma}\left(A;B\mid C\right)=I_{\sigma}\left(B;A\mid C\right)\,.\\
\textrm{\textbf{Chain rule}} & I_{\sigma}\left(A;BC\mid D\right)=I_{\sigma}\left(A;B\mid D\right)\\
 & \,\,\,\,\,\,\,\,\,\,\,\,\,\,\,\,\,\,\,\,\,\,\,\,\,\,\,\,\,\,\,\,\,\,\,\,\,\,\,\,\,\,\,\,\,\,+I_{\sigma}\left(A;C\mid BD\right)\,.
\end{array}\label{eq:Chain}
\end{equation}
\end{definition}
\begin{theorem}
Assume that $D_{F}$ is a monotone and additive Bregman divergence.
Then conditional mutual information is a separoid function.
\end{theorem}
\begin{svmultproof}
\textbf{Positivity} Conditional mutual information can be rewritten
as 
\begin{dmath*}
I_{\sigma}\left(A;B\mid C\right)=D_{F}\left(\sigma_{ABC},\sigma_{A}\otimes\sigma_{B}\otimes\sigma_{C}\right)\\
-D_{F}\left(\sigma_{AC},\sigma_{A}\otimes\sigma_{C}\right)-D_{F}\left(\sigma_{BC},\sigma_{B}\otimes\sigma_{C}\right)\\
=D_{F}\left(\sigma_{ABC},\sigma_{B}\otimes\sigma_{AC}\right)-D_{F}\left(\sigma_{BC},\sigma_{B}\otimes\sigma_{C}\right)\\
=D_{F}\left(\sigma_{ABC},\sigma_{B}\otimes\sigma_{AC}\right)\\
-D_{F}\left(\sigma_{A}\otimes\sigma_{BC},\sigma_{A}\otimes\sigma_{B}\otimes\sigma_{C}\right).
\end{dmath*}
Let $\Phi$ denote the affinity $\Phi\left(\rho\right)=\sigma_{A}\otimes E_{BC}\left(\rho\right)$.
Then
\begin{align*}
\Phi\left(\sigma_{ABC}\right) & =\sigma_{A}\otimes\sigma_{BC}\,,\\
\Phi\left(\sigma_{B}\otimes\sigma_{AC}\right) & =\sigma_{A}\otimes\sigma_{B}\otimes\sigma_{C}\,.
\end{align*}
Therefore monotonicity implies that $I_{\sigma}\left(A;B\mid C\right)$
cannot be negative.

\textbf{Symmetry} It follows directly from the definition that conditional
mutual information is symmetric.

\textbf{Chain rule} To prove the chain rule we expand the left hand
side of Equation (\ref{eq:Chain}) as 
\begin{multline*}
I_{\sigma}\left(A;BC\mid D\right)=D_{F}\left(\sigma_{ABCD},\sigma_{A}\otimes\sigma_{BC}\otimes\sigma_{D}\right)\\
-D_{F}\left(\sigma_{AD},\sigma_{A}\otimes\sigma_{D}\right)-D_{F}\left(\sigma_{BCD},\sigma_{BC}\otimes\sigma_{D}\right).
\end{multline*}
Next we use Equation (\ref{eq:marginal}) to get
\begin{align*}
I_{\sigma}\left(A;BC\mid D\right)&\\
=&\left(\begin{array}{c}
D_{F}\left(\sigma_{ABCD},\sigma_{A}\otimes\sigma_{B}\otimes\sigma_{C}\otimes\sigma_{D}\right)\\
-D_{F}\left(\sigma_{BC},\sigma_{B}\otimes\sigma_{C}\right)
\end{array}\right)\\
&-D_{F}\left(\sigma_{AD},\sigma_{A}\otimes\sigma_{D}\right)\\
&-\left(\begin{array}{c}
D_{F}\left(\sigma_{BCD},\sigma_{B}\otimes\sigma_{C}\otimes\sigma_{D}\right)\\
-D_{F}\left(\sigma_{BC},\sigma_{B}\otimes\sigma_{C}\right)
\end{array}\right).
\end{align*}
The left hand side reduces to
\begin{multline}
I_{\sigma}\left(A;BC\mid D\right)=D_{F}\left(\sigma_{ABCD},\sigma_{A}\otimes\sigma_{B}\otimes\sigma_{C}\otimes\sigma_{D}\right)\\
-D_{F}\left(\sigma_{AD},\sigma_{A}\otimes\sigma_{D}\right)-D_{F}\left(\sigma_{BCD},\sigma_{B}\otimes\sigma_{C}\otimes\sigma_{D}\right).\label{eq:leftKaede}
\end{multline}
Similarly, we expand the right hand side of Equation (\ref{eq:Chain})
as
\begin{multline*}
I_{\sigma}\left(A;B\mid D\right)+I_{\sigma}\left(A;C\mid BD\right)\\
=D_{F}\left(\sigma_{ABD},\sigma_{A}\otimes\sigma_{B}\otimes\sigma_{D}\right)-D_{F}\left(\sigma_{AD},\sigma_{A}\otimes\sigma_{D}\right)\\
-D_{F}\left(\sigma_{BD},\sigma_{B}\otimes\sigma_{D}\right)+D_{F}\left(\sigma_{ABCD},\sigma_{A}\otimes\sigma_{C}\otimes\sigma_{BD}\right)\\
-D_{F}\left(\sigma_{ABD},\sigma_{A}\otimes\sigma_{BD}\right)-D_{F}\left(\sigma_{BCD},\sigma_{C}\otimes\sigma_{BD}\right)\,.
\end{multline*}
We use Equation (\ref{eq:marginal}) to re-write the three last terms
as
\begin{multline*}
I_{\sigma}\left(A;B\mid D\right)+I_{\sigma}\left(A;C\mid BD\right)\\
=D_{F}\left(\sigma_{ABD},\sigma_{A}\otimes\sigma_{B}\otimes\sigma_{D}\right)-D_{F}\left(\sigma_{AD},\sigma_{A}\otimes\sigma_{D}\right)\\
\,\,\,\,-D_{F}\left(\sigma_{BD},\sigma_{B}\otimes\sigma_{D}\right)\\
\,\,\,\,+\left(\begin{array}{c}
D_{F}\left(\sigma_{ABCD},\sigma_{A}\otimes\sigma_{B}\otimes\sigma_{C}\otimes\sigma_{D}\right)\\
-D_{F}\left(\sigma_{BD},\sigma_{B}\otimes\sigma_{D}\right)
\end{array}\right)\\
\,\,\,\,-\left(\begin{array}{c}
D_{F}\left(\sigma_{ABD},\sigma_{A}\otimes\sigma_{B}\otimes\sigma_{D}\right)\\
-D_{F}\left(\sigma_{BD},\sigma_{B}\otimes\sigma_{D}\right)
\end{array}\right)\\
-\left(\begin{array}{c}
D_{F}\left(\sigma_{BCD},\sigma_{B}\otimes\sigma_{C}\otimes\sigma_{D}\right)\\
-D_{F}\left(\sigma_{BD},\sigma_{B}\otimes\sigma_{D}\right)
\end{array}\right)\,.
\end{multline*}
The right hand side reduces to
\begin{multline}
I_{\sigma}\left(A;B\mid D\right)+I_{\sigma}\left(A;C\mid BD\right)\\
=D_{F}\left(\sigma_{ABCD},\sigma_{A}\otimes\sigma_{B}\otimes\sigma_{C}\otimes\sigma_{D}\right)\\
\,\,\,\,-D_{F}\left(\sigma_{AD},\sigma_{A}\otimes\sigma_{D}\right)-D_{F}\left(\sigma_{BCD},\sigma_{B}\otimes\sigma_{C}\otimes\sigma_{D}\right).\label{eq:rightKaede}
\end{multline}
Since the left hand side (\ref{eq:leftKaede}) and the right hand
side (\ref{eq:rightKaede}) are equal we have proved the chain rule
(\ref{eq:Chain}).
\end{svmultproof}

\section{Conslusion\label{sec:conclusion}}

We have carefully described concepts like state space and introduced state spaces 
on Jordan algebras as the most important example. In general probabilistic theories there are different ways of defining the entropy of a state, 
but these different definitions coincide on Jordan algebras.
For any optimization problem trere is an associated Bregman divergence, but with extre constraints like monotonicity, sufficiency, or statistical locality 
a Bregman divergenceon a Jordan algebra is proportional to the Bregman divergence generated by the uniquely defined entropy function. 
A monotone Bregman divergence on a Jordan algebra is automatically additive. For composed systems an additive and monotone  Bregman divergence can be used to define conditional mutual information and this quantity will satisfy consistency,
the data processing inequality and the chain rule. In \cite{Pawlowski2009} it was proved that if conditional mutual information can be defined in a way such that 
consistency,
the data processing inequality and the chain rule are satisfied then the system will satisfy the condition called 
\emph{information causality} \cite{Pawlowski2009}. In \cite{Pawlowski2009} it was also proved that a system that satisfies
information causality cannot have super-quantum correlations, i.e.
correlations violate Tsirelson's bound. The conclusion is that the existence of a monotone Bregman divergence implies that super-quantum correlations do not exist.

The results work out nicely on Jordan algebras, but maybe it will work in any generalized probabilistic theory. For instance it would be interesting if the following conjecture holds.

\begin{conjecture}
All monotone Bregman divergences are additive. 
\end{conjecture}

A careful inspection of the proofs also reveal that the results involving Jordan algebras only involve that the cone is self dual and that a Euclidean Jordan algebra is {\em strongly spectral} in the sense that $f(\sigma)$ is well defined for any function $f$. Appearently monotonicity of a Bregman divergence implies spectrality, but the only solid result in this direction is the following theorem.

\begin{theorem}[\cite{Harremoes2017c}]
If a state space has rank 2 and it has a strict and monotone Bregman
divergence then the state space can be represented as a spin factor. 
In particular the state spce is strongly spectral.
\end{theorem}

For most convex bodies it is not possible to define a monotone Bregman
divergence and it is not known if it is possible to define a monotone
Bregman divergences on any convex body that cannot be represented
by a Jordan algebra. It would be highly desirable to classify state
spaces with monotone Bregman divergences in cases when the rank exceeds
2.

\section*{Conflict of interest}
The corresponding author states that there is no conflict of interest.


\end{document}